\documentclass[10pt,a4paper,reqno]{amsart}
\usepackage{acronym}
\usepackage{amsfonts}
\usepackage{amsmath}
\usepackage{amssymb}
\usepackage{amsthm}
\usepackage{bm}
\usepackage{braket}
\usepackage{cite}
\usepackage{color}
\usepackage{geometry}
\usepackage{graphicx}
\usepackage{hyperref}
\usepackage{mathrsfs}
\usepackage{mathtools}
\usepackage{setspace}
\usepackage{stmaryrd}
\usepackage{subcaption}
\usepackage{tikz}

\usepackage{epsfig}
\usepackage{setspace}
\usepackage{booktabs}
\usepackage{threeparttable}
\usepackage{diagbox}

\newgeometry{top=2cm,bottom=2cm,outer=1.5cm,inner=1.5cm}
\newcommand{\bse}{\begin{subequations}}
\newcommand{\ese}{\end{subequations}}

\newtheorem{theorem}{Theorem}

\newtheorem{lemma}[theorem]{Lemma}

\newtheorem{proposition}[theorem]{Proposition}
\newtheorem{rhp}[theorem]{Riemann-Hilbert Problem}
\numberwithin{equation}{section}

\title[New form soliton solutions and multiple zeros solutions for higher-order KN equation]{A new form of general soliton solutions and multiple zeros solutions for a higher-order Kaup-Newell equation}

\author{JinYan Zhu}
\address[JY]{School of Mathematical Sciences, Shanghai Key Laboratory of Pure Mathematics and Mathematical Practice\\
East China Normal University \\ Shanghai 200241 \\ People's Republic of China}
\author{Yong Chen$^*$}
\address[YC]{School of Mathematical Sciences, Shanghai Key Laboratory of Pure Mathematics and Mathematical Practice \\
East China Normal University \\ Shanghai 200241 \\ People's Republic of China}
\address[YC]{College of Mathematics and Systems Science \\ Shandong University of Science and Technology \\ Qingdao 266590 \\ People's Republic of China}
\email{ychen@sei.ecnu.edu.cn}
\begin{document}

\begin{abstract}
Due to higher-order Kaup-Newell (KN) system has more complex and diverse solutions than classical second-order flow KN system, the research on it has attracted more and more attention. In this paper, we consider a higher-order KN equation with third order dispersion and quintic nonlinearity. Based on the theory of  the inverse scattering, the matrix Riemann-Hilbert problem is established. Through the dressing method, the solution matrix with simple zeros without reflection is constructed. In particular, a new form of solution is given, which is more direct and simpler than previous methods. In addition, through the determinant solution matrix, the vivid diagrams and dynamic analysis of single soliton solution and two soliton solution are given in detail. Finally,  by using the technique of limit, we construct the general solution matrix in the case of multiple zeros, and the examples of solutions for the cases of double zeros, triple zeros, single-double zeros and double-double zeros are especially shown.
\end{abstract}

\maketitle

\section{Introduction}
In \cite{KA-2018-JMP}, Abhinav et al. give an coupled equations
\begin{equation}
\begin{array}{l}
q_{t}=i q_{x x}-(4 \beta+1) q^{2} r_{x}-4 \beta q q_{x} r+\frac{i}{2}(1+2 \beta)(4 \beta+1) q^{3} r^{2}, \\
r_{t}=-i r_{x x}-(4 \beta+1) r^{2} q_{x}-4 \beta r r_{x} q-\frac{i}{2}(1+2 \beta)(4 \beta+1) q^{2} r^{3}.\label{KD}
\end{array}
\end{equation}
The system (\ref{KD}) has three famous Schr\"{o}dinger-type reductions and these three reductions had been widely studied in recent years.

When $\beta=-\frac{1}{2}$ and $r=-q^{*}$ the system (\ref{KD}) reduce to DNLS I
\begin{equation}
i q_{t}+q_{x x}-i\left(q^{2} q^{*}\right)_{x}=0,\label{DNLS1}
\end{equation}
the symbol $'*'$ represents the complex conjugate, and the subscript of $x$ (or $t$) represents the partial derivative with respect to $x$ (or $t$). Eq. (\ref{DNLS1}) is also called the Kaup-Newell (KN) equation \cite{Kaup-1978-JMP}. In recent years, the KN equation related to spectral problems, exact solutions, Hamilton structure, Painl\'{e}ve properties and other properties have been in-depth research \cite{Kaup-1978-JMP,EG-2001-PA,HA-2021-AMP,ZYB-1994-PD,WY-2018-RRP,JL-2011-PD,HS-2003-JPA}. Eq. (\ref{DNLS1}) is a typical dispersion equation, which is derived from the magnetohydrodynamic equation with Hall effect, especially describing the nonlinear Alfv\'{e}n waves in plasma physics \cite{KM-1976-PS,EM-1989-PS,AB-2018-OP}. 

When $\beta=-\frac{1}{4}$ and $r=-q^{*}$ the system (\ref{KD}) reduce to DNLS II
 \begin{equation}
i q_{t}+q_{x x}-i q q^{*} q_{x}=0,\label{CLL}
\end{equation}
which appears in optical models of ultrashort pulses and is also referred to as the Chen-Lee-Liu (CLL) equation \cite{CLL-1979-PS}.

When $\beta=0$ and $r=-q^{*}$ the system (\ref{KD}) reduce to DNLS III
\begin{equation}
i q_{t}+q_{x x}-i q^{2} q_{x}^{*}+\frac{1}{2} q^{3} q^{* 2}=0,\label{GI}
\end{equation}
 The last equation is first found by Gerdjikov and Ivanov in \cite{GI-1983-BJP} also known as the GI equation.

In \cite{EG-2001-JPA}, Fan. give the higher-order generalization of (\ref{KD}) equations
\begin{equation}
\begin{aligned}
q_{t}-&\frac{1}{4}[2 q_{x x x}-6(2 \beta-1) r q_{x}^{2}-6(4 \beta-1) q q_{x} r_{x}-6(2 \beta-1) q r q_{x x}+6(2 \beta-1)(4 \beta-1) q^{3} r r_{x}\\
+&3(8 \beta^{2}-12 \beta+3) q^{2} r^{2} q_{x}+4 \beta(2 \beta-1)(4 \beta-1) q^{4} r^{3}]=0 \\
r_{t}-&\frac{1}{4}[2 r_{x x x}+ 6(2 \beta-1) q r_{x}^{2}-6(4 \beta-1) r q_{x} r_{x}+6(2 \beta-1) q r r_{x x}+6(2 \beta-1)(4 \beta-1) q r^{3} q_{x}\\
+&3(8 \beta^{2}-12 \beta+3) q^{2} r^{2} r_{x}
-4 \beta(2 \beta-1)(4 \beta-1) q^{3} r^{4}]=0.
\end{aligned}\label{Fan2}
\end{equation}
The system (\ref{Fan2}) can be used to describe the  higher-order nonlinear effects in nonlinear optics and other fields.
Eq. (\ref{Fan2}) also has three important Schr\"{o}dinger-type reductions.

First, when $\beta=0, x \rightarrow i x, t \rightarrow i t$ and $r=-q^{*}$ the system (\ref{Fan2}) become
\begin{equation}
q_t=-\frac{1}{2}q_{xxx}+(\frac{3i}{2}|q|^2q_{x})_x+(\frac{3}{4}|q|^4q)_x,\label{hDNLS1}
\end{equation}
which can be viewed as the higher-order DNLS I or higher-order KN equation.  Eq.(\ref{hDNLS1}) also can be derived from the generalized KN hierarchy \cite{GSF-2013-JPA} under $n = 3$ and proper parameter.

Second, when $\beta=\frac{1}{4}$ and $ x \rightarrow i x, t \rightarrow i t, r=-q^{*}$ the  system (\ref{Fan2}) become
\begin{equation}
q_t=-\frac{1}{2}q_{xxx}-\frac{3}{4}i|q|^2q_{xx}-\frac{3}{4}iq^{*}q_{x}^{2}+\frac{3}{8}|q|^{4}q_{x},\label{hDNLS2}
\end{equation}
which can be viewed as the higher-order DNLS II or higher-order CLL equation.

Third, when $\beta=\frac{1}{2}$ and $ x \rightarrow i x, t \rightarrow i t$ and $r=-q^{*}$, the  system (\ref{Fan2}) become
 \begin{equation}
q_t=-\frac{1}{2}q_{xxx}+\frac{3}{2}iqq_{x}q_{x}^{*}-\frac{3}{4}|q|^4q_{x}.\label{hDNLS3}
\end{equation}
which can be regarded as the higher-order DNLS III or higher-order GI equation.

 It has been proved in \cite{EG-2001-JPA} that these equations (\ref{hDNLS1})-(\ref{hDNLS3})have multiple Hamiltonian structures and are Liouville integrable. The N-soliton solutions of Eq.(\ref{hDNLS2}) and Eq.(\ref{hDNLS3}) have been studied in \cite{JH-2018-JNMP,ZJ-2020-ar}. In this paper, we mainly consider the soliton solutions and higher-order soliton solution of system (\ref{hDNLS1}).
 In fact, there are several classical methods to obtain the soliton solutions, such as the inverse scattering (IST) method, Darboux/B\"{a}cklund transform, Hirota bilinear method, RH method\cite{Kundu-1987-PD,KI-1999-JPSJ,GB-2012-JMP,LL-2010-JMP,MA-1974-MP,XJ-2004-PRE}. Here we will use the Riemann-Hilbert(RH) method to derive the soliton solutions of (\ref{hDNLS1}) since it is more convenient  to  study the exact long-time asymptotic and large $n$ asymptotic\cite{LN-2019-ar}.

The high-order soliton solution of the NLS type has been widely concerned by many scholars in recent years. It can be used to describe the weak bound states of solitons, which may appear in the study of soliton train transmission with specific chirp and almost equal velocity and amplitude\cite{LG-1994-OL}. There are not many studies on DNLS type higher-order soliton solutions. Recently, Chen's team studied the double and triple zeros of GI equation\cite{WQ-2021-ar}, and the double zeros of higher-order KN\cite{JP-2021-ar}. Here, we study more extensive cases and give the general form of the solutions with multiple zeros.

The main content of this paper is to construct the general soliton solution matrix of the higher-order KN equation by using RH method. It is worth noting that we recover the potential $q(x,t)$ as the spectral parameter $\zeta\rightarrow0$, it effectively reduces the operation process and avoids the interference of implicit function, and the matrix form of the soliton solution is more direct. Taking the single soliton solution and the two-soliton solution as examples, the properties of the soliton are studied. Then, on the basis of the soliton solution, through a certain limit technique, the solution matrix of the high-order soliton solution of the multiple zeros is obtained.

The organization of this letter is as follows. In section 2, the inverse scattering theory is established for the $2 \times 2$ spectral problems and the corresponding matrix Riemann-Hilbert problem (RHP) is formulated. The N-soliton formula for higher-order KN equation is derived by considering the simple zeros in the RHP in section 3. In section 4, we construct the higher-order soliton matrix and obtain the general expression of the higher-order soliton, which corresponds to the multiple zeros in the RHP. The section 5 is devoted to conclusion and discussion.

\section{Inverse scattering  theory of (\ref{hDNLS1})}
The main work of this part is to study the inverse scattering problem of Eq.(\ref{hDNLS1}) and construct the corresponding RHP.

The Eq. (\ref{hDNLS1}) is Lax integrable with the linear spectral problem
\begin{equation}
Y_{x}=MY, ~~~~~~M=-i\zeta^2 \sigma_3+\zeta Q,\label{x part}
\end{equation}
\begin{equation}
Y_{t}=NY,~~~~~~N=-2i\zeta^6 \sigma_3+N_1,\label{t part}
\end{equation}
where
\begin{equation}\begin{split}
N_{1}= &2 Q \zeta^{5}-i Q^{2} \sigma_{3} \zeta^{4}+i\sigma_{3} Q_x \zeta^{3}+Q^{3} \zeta^{3} -\frac{1}{2}\left(Q Q_{x}-Q_{x} Q\right)  \zeta^{2} \\
&-\frac{3}{4} i Q^{4} \sigma_{3} \zeta^{2}-\frac{1}{2} Q_{x x} \zeta+\frac{3 i}{2} \sigma_{3} Q^{2} Q_{x} \zeta+\frac{3}{4} Q^{5} \zeta,\end{split}\label{5.8}\end{equation}
\begin{equation}
\sigma_3=\left(\begin{array}{cc}
1 & 0 \\
0 & -1
\end{array}\right),~~~Q=\left(\begin{array}{cc}
0 & q \\
-q^{*} & 0
\end{array}\right),\label{5.12}
\end{equation}
it is easy to verify
$$Q^{\dagger} = -Q, ~~~~~~~\sigma_{3} Q \sigma_{3}=-Q,$$
which plays an important role in symmetry research later, and the symbol $'\dagger'$ represents the conjugate transpose of a matrix.
In the following analysis, we assume that the potential function $q, q^{*}$ rapidly tends to zero as $x\rightarrow\pm\infty$. In this case, the solution of the boundary form can be clearly obtained
\begin{equation}Y \sim e^{(-i \zeta^2 x-2 i \zeta^{6} t)\sigma_{3}}, \text{as}~~~~x\rightarrow \infty.\label{2.17}\end{equation}

Making the following transformation
\begin{equation}
Y=Je^{(-i \zeta^2 x-2 i \zeta^{6} t)\sigma_{3}},\label{yj}
\end{equation}
The Lax pair of Eq. (\ref{x part})-(\ref{t part}) becomes
\begin{equation}
J_{x}+\mathrm{i} \zeta^{2}\left[\sigma_{3},J\right]=\zeta Q J, \label{Jx}
\end{equation}
\begin{equation}
J_{t}+2i\zeta^{6}\left[\sigma_{3}, J\right]=N_1 J,\label{Jt}
\end{equation}
where $Q,N_1$ has been given by  Eq.(\ref{5.8}),(\ref{5.12}).

In the scattering problem, the Lax equation (\ref{Jt}) of time $t$ is ignored temporarily.
 By solving Eq.(\ref{x part}) with constant variation method and using transformation (\ref{yj}), the solution of Eq.(\ref{Jx}) can be obtained, which satisfies the following integral equations
\begin{equation}
J_{M}=I+ \zeta \int_{-\infty}^{x} e^{\mathrm{i} \zeta^{2} \sigma_{3}(y-x)} Q(y) J_{M} e^{\mathrm{i} \zeta^{2} \sigma_{3}(x-y)} \mathrm{d} y,\label{J-1}
\end{equation}
\begin{equation}
J_{P}=I- \zeta \int_{x}^{+\infty} e^{\mathrm{i} \zeta^{2} \sigma_{3}(y-x)} Q(y) J_{P} e^{\mathrm{i} \zeta^{2} \sigma_{3}(x-y)} \mathrm{d} y,\label{J-2}
\end{equation}
and these two Jost solutions are satisfied the following asymptotics at large distances
\begin{equation}
J\sim I,~~as~|x|\sim\infty.\label{ii}
\end{equation}

In order to analyze the analytical properties of Jost solutions in the $\zeta$ plane,  we divide the entire $\zeta$ plane into two regions
$$
\mathbb{C}_{13}=\left\{\zeta \mid \arg \zeta \in\left(0, \frac{\pi}{2}\right) \cup\left(\pi, \frac{3 \pi}{2}\right)\right\}, \quad \mathbb{C}_{24}=\left\{\zeta \mid \arg \zeta \in\left(\frac{\pi}{2}, \pi\right) \cup\left(\frac{3 \pi}{2}, 2 \pi\right)\right\}.
$$
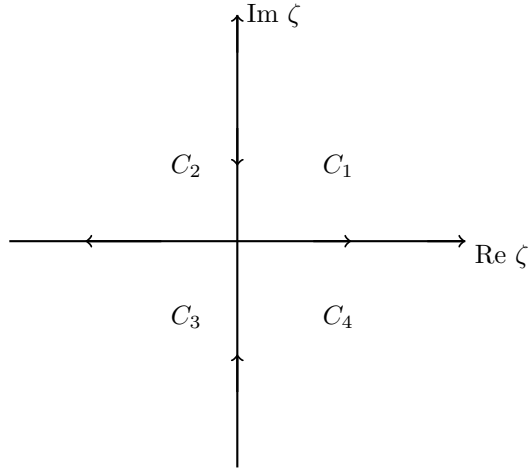
\begin{figure}[htpb]
\center
\begin{tikzpicture}\usetikzlibrary{arrows}
\coordinate [label=0: Im $\zeta$] ()at (0,3);
\coordinate [label=0: Re $\zeta$] ()at (3,-0.2);
\coordinate [label=0:] ()at (2,0.1);
\coordinate [label=0:] ()at (-2.6,0.1);
\coordinate [label=0: $C_{1}$] ()at (1,1);
\coordinate [label=0: $C_{2}$] ()at (-1,1);
\coordinate [label=0: $C_{3}$] ()at (-1,-1);
\coordinate [label=0: $C_{4}$] ()at (1,-1);
\draw[->, thick] [fill](1,0)--(1.5,0);
\draw[->, thick] (2.5,0)--(3,0);
\draw[->,thick ] (-1,0)--(-2,0);
\draw[->, thick] (0,1.5)--(0,1);
\draw[->,thick] (0,2.5)--(0,3);
\draw[->,thick] (0,-2)--(0,-1.5);
\draw[thick](0,3)--(0,-3);
\draw[thick](3,0)--(-3,0);
\end{tikzpicture}
\caption{ Definition of the $\mathbb{C}_{13}$=$C_{1}\cup C_{3}$ and $\mathbb{C}_{24}=C_{2}\cup C_{4}$}
\label{hdnls-fig.1}
\end{figure}

Dividing $J$ into columns as $J=\left(J^{(1)}, J^{(2)}\right)$, due to the structure (\ref{5.12}) of the potential $Q$, and Volterra integral equations (\ref{J-1})-(\ref{J-2}), we have
\begin{proposition}
The above Volterra integral equations exist and are unique, and have the following properties:

$\bullet$ The column vectors $J_{M}^{(1)}$ and $J_{P}^{(2)}$ are continuous for $\zeta \in \mathbb{C}_{13} \cup \mathbb{R} \cup i \mathbb{R}$ and analytic for $\zeta \in \mathbb{C}_{13}$,

$\bullet$The column vectors $J_{p}^{(1)}$ and $J_{M}^{(2)}$ are continuous for $\zeta \in \mathbb{C}_{24} \cup \mathbb{R} \cup i\mathbb{R}$ and analytical for $\zeta \in \mathbb{C}_{24}$.
\end{proposition}

Through the Eq.(\ref{yj}) we know that $J_{P} E$ and $J_{M} E$ are both solutions of the linear Eq. (\ref{x part}), so they are linearly related by a matrix $S(\zeta)$
\begin{equation}
J_{M} E=J_{P} E S(\zeta), \quad \zeta \in \mathbb{R} \cup \mathrm{i} \mathbb{R},\label{2.21}
\end{equation}
where $E=e^ {-\mathrm{i} \zeta^{2} x \sigma_{3}}$ and $S(\zeta)=\left(s_{i j}\right)_{2 \times 2}$. It should be noted that
$$
\text{tr}(-i\zeta^2 \sigma_3+\zeta Q)=0,
$$
using the Abel's  formula, we can get that
\begin{equation}
(\operatorname{det} Y)_{x}=0,\label{y_x}
\end{equation}
considering transformation (\ref{yj}) has
$$
\operatorname{det} J=\operatorname{det} Y \operatorname{det}(e^{\mathrm{i} \zeta^{2} x \sigma_{3}})=\operatorname{det} Y.
$$
Reusing Eq.(\ref{y_x}) has
$$
\left(\operatorname{det} J\right)_{x}=0,
$$
which means that the $\text{det} J$ is independent of $x$, and then from the asymptotic (\ref{ii}), we know
$$
\operatorname{det} J=\lim _{|x| \rightarrow \infty} \operatorname{det} J=\operatorname{det}(\lim _{|x| \rightarrow \infty} J)=1.
$$
Taking the determinant on both sides of relation (\ref{2.21}) to get
$\operatorname{det} S(\lambda)=1$.

In order to construct the RHP, we consider the adjoint scattering equation of (\ref{Jx})
\begin{equation}
\Phi_{x}=-i\zeta^2\left[\sigma_{3}, \Phi\right]-\zeta\Phi Q,\label{adeq}
\end{equation}
it is easy to see that $J^{-1}$ is the solution of the adjoint equation (\ref{adeq}) and satisfy the boundary condition $J^{-1} \rightarrow I$ as $x \rightarrow \pm \infty,$ where the inverse matrices $J^{-1}$ as a
collection of rows
\begin{equation}
\left(J_{P}\right)^{-1}=((J_{P}^{-1})^{(1)},(J_{P}^{-1})^{(2)})^{T}, \quad \left(J_{M}\right)^{-1}=((J_{M}^{-1})^{(1)},(J_{M}^{-1})^{(2)})^{T}.
\end{equation}
Due to the structure (\ref{5.12}) of the potential $Q$, we also have
\begin{proposition}
According to the properties of Jost solution, we can deduce that the inverse matrix $J^{-1}$ has the following properties:

$\bullet$ The row vectors $(J_{P}^{-1})^{(1)}$ and $(J_{M}^{-1})^{(2)}$ are continuous for $\zeta \in \mathbb{C}_{13} \cup \mathbb{R} \cup i \mathbb{R}$ and analytic for $\zeta \in \mathbb{C}_{13}$.

$\bullet$ The rows $(J_{M}^{-1})^{(1)}$ and $(J_{P}^{-1})^{(2)}$ are continuous for $\zeta \in \mathbb{C}_{24} \cup \mathbb{R} \cup i\mathbb{R}$ and analytical for $\zeta \in \mathbb{C}_{24}$.
\end{proposition}
Further, the analytical properties of the scattering data can be obtained as follows:
\begin{proposition}
Suppose that $q(x, t) \in L^1(\mathbb{R})$, then $s_{11}$ is analytic on $\mathbb{C}_{13}$, $s_{22}$  is analytic on $\mathbb{C}_{24}$; $s_{12}$ and $s_{22}$ are not analytic in $\mathbb{C}_{13}$ and $\mathbb{C}_{24}$, but are continuous to the real axis $\mathbb{R}$ and imaginary axis $\mathrm{i} \mathbb{R}$.
\end{proposition}
\begin{proof}
 The scattering matrix can be rewritten as:
\begin{equation}
e^ {-\mathrm{i} \zeta^{2} x \sigma_{3}} S(\zeta)e^ {\mathrm{i} \zeta^{2} x \sigma_{3}}=J_P^{-1}J_M
=\left(\begin{array}{c}
\left(J_{P}^{-1}\right)^{(1)} \\ \left(J_{P}^{-1}\right)^{(2)}
\end{array}\right)\left(J_{M}^{(1)}, J_{M}^{(2)}\right)=\left(\begin{array}{cc}
\left(J_{P}^{-1}\right)^{(1)} J_{M}^{(1)} & \left(J_{P}^{-1}\right)^{(1)} J_{M}^{(2)}   \\
\left(J_{P}^{-1}\right)^{(2)} J_{M}^{(1)} & \left(J_{P}^{-1}\right)^{(2)} J_{M}^{(2)}
\end{array}\right), \quad \zeta \in \mathbb{R} \cup \mathrm{i} \mathbb{R}.
\end{equation}
The elements corresponding to the matrices on both sides can be written clearly
$$
s_{11}=\left(J_{P}^{-1}\right)^{(1)} J_{M}^{(1)},~~~s_{12}=\left(J_{P}^{-1}\right)^{(1)} J_{M}^{(2)}e^{2\mathrm{i} \zeta^{2} x},
$$
$$
s_{21}=\left(J_{P}^{-1}\right)^{(2)} J_{M}^{(1)}e^{-2\mathrm{i} \zeta^{2} x},~~~s_{12}=\left(J_{P}^{-1}\right)^{(2)} J_{M}^{(2)}.
$$
According to proposition 1 and Proposition 2, it's easy to know $s_{11}$ is analytic on $\mathbb{C}_{13}$, $s_{22}$  is analytic on $\mathbb{C}_{24}$; $s_{12}$ and $s_{22}$ are not analytic in $\mathbb{C}_{13}$ and $\mathbb{C}_{24}$, but are continuous to the real axis $\mathbb{R}$ and imaginary axis $\mathrm{i} \mathbb{R}$.
\end{proof}

Hence,  we can construct two matrix functions $\mathbf{P}(\zeta, x)$ which are analytic for $\zeta\in\mathbb{C}_{13}\cup\mathbb{C}_{24}$,
\begin{equation}
\mathbf{P}(\zeta, x):=\left\{
\begin{aligned}&\left[(J_{M}^{(1)}(\zeta, x), J_{P}^{(2)}(\zeta, x)\right],\,\,\, \zeta\in \mathbb{C}_{13}\\
&\left[ (J_{M}^{-1})^{(1)}(\zeta, x), (J_{P}^{-1})^{(2)}(\zeta, x)\right],\,\,\, \zeta\in \mathbb{C}_{24}
\end{aligned}
\right.
\end{equation}
and $\text{det}\mathbf{P}=s_{11}$, when $\zeta\in \mathbb{C}_{13}$, $\text{det}\mathbf{P}=\hat{s}_{11}$, when $\zeta\in \mathbb{C}_{24}$.
$\hat{s}_{11}$ is the first element of $S^{-1}$.

To find the boundary condition of $\mathbf{P}$, we consider the following asymptotic expansion as $\zeta \rightarrow 0,$
\begin{equation}
\mathbf{P}=\mathbf{P}^{(0)}+\mathbf{ P}^{(1)}\zeta+ \mathbf{P}^{(2)}\zeta^{2}+O(\zeta^{3}).\label{2.23}
\end{equation}
Substituting (\ref{2.23}) into (\ref{Jx}) and equating terms with like powers of $\zeta$, we find
$$
\mathbf{P}^{(0)}_{x}=0.
$$
It can be seen from (\ref{J-1}) and (\ref{J-2})
$$
J|_{(\zeta=0)}=I,
$$
so we have
\begin{equation}
\mathbf{P} \rightarrow I,~~~\zeta\rightarrow 0.\label{2.24}
\end{equation}
Then the Riemann-Hilbert problem of the higher-order KN equation is
\begin{rhp}
The matrix function $\mathbf{P}(\zeta; x)$ has the following properties:

$\bullet$ $\mathbf{Analyticity}: \mathbf{P}(\zeta; x, t)$ is analytic function in $\zeta\in\mathbb{C}_{13}\cup\mathbb{C}_{24}$;

$\bullet$  $\mathbf{Jump\quad Condition}$: \begin{equation} \mathbf{P}_{+}(\zeta; x)=\mathbf{P}_{-}(\zeta;x)G(\zeta),~~\quad \zeta \in \mathbb{R} \cup \mathrm{i} \mathbb{R}.\label{RHP}
    \end{equation}

$\bullet$  $\mathbf{Normalization}: \mathbf{P}(\zeta; x)=I+O(\zeta),~~~~~\text{as } \zeta\rightarrow 0.$
\end{rhp}
Where
\begin{equation}
G=E\left(\begin{array}{cc}
1 & \hat{s}_{12}\\
s_{21} & 1  \\
\end{array}\right) E^{-1}.\label{G1}
\end{equation}

Next, we consider the symmetric properties of Jost solutions and scattering data, so that we can consider interesting reduction.
\begin{proposition}
There are two symmetries of the Jost solutions and scattering data:

$\bullet$ The first symmetry reduction
\begin{equation}(J(x, \zeta^{*}))^{\dagger}=J^{-1}(x, \zeta),\label{2.28}\end{equation}
\begin{equation}(\mathbf{P}(\zeta^{*}))^{\dagger}=\mathbf{P}^{-1}(\zeta),\label{2.29}\end{equation}
\begin{equation}S^{\dagger}\left(\zeta^{*}\right)=S^{-1}(\zeta).\label{2.30}\end{equation}
$\bullet$ The second symmetry reduction
\begin{equation}
J(\zeta)=\sigma_{3} J(-\zeta) \sigma_{3},\label{2.31}
\end{equation}
\begin{equation}
\mathbf{P}(-\zeta)=\sigma_{3} \mathbf{P}(\zeta) \sigma_{3},\label{2.32}
\end{equation}
\begin{equation}
S(-\zeta)=\sigma_{3} S(\zeta) \sigma_{3}.\label{2.33}
\end{equation}
\end{proposition}
\begin{proof}
For the first symmetric case, replacing $\zeta$ with $\zeta^{*}$, and then take the conjugate transpose of the Eq.(\ref{Jx}) to get
\begin{equation}
(J^{\dagger}(x,\zeta^{*}))_{x}=-\mathrm{i} \zeta^{2}[\sigma_{3}, J^{\dagger}(x,\zeta^{*})]+\zeta J^{\dagger}(x,\zeta^{*}) Q^{\dagger},
\end{equation}
owing to $Q^{\dagger} = -Q$, so the above equation is
\begin{equation}
(J^{\dagger}(x,\zeta^{*}))_{x}=-\mathrm{i} \zeta^{2}[\sigma_{3}, J^{\dagger}(x,\zeta^{*})]- \zeta J^{\dagger}(x,\zeta^{*}) Q,\label{2.27}
\end{equation}
Comparing with Eq.(\ref{adeq}), it is found that $J^{-1}(x,\zeta)$ and $J^{\dagger}(x,\zeta^{*})$ satisfy the same equation form, and then according to the boundary conditions at $x \rightarrow \pm \infty$, we know that
$$(J(x, \zeta^{*}))^{\dagger}=J^{-1}(x, \zeta).$$
Notice that the $\mathbf{P}$ that we construct is part of the Jost solutions, so there must be also
$$
(\mathbf{P}(\zeta^{*}))^{\dagger}=\mathbf{P}^{-1}(\zeta).
$$
In addition, in view of the scattering relation (\ref{2.21}) between $J_{M}$ and $J_{P}$, we see that $S$ also satisfies the involution property
$$S^{\dagger}\left(\zeta^{*}\right)=S^{-1}(\zeta).$$

For the second symmetry, replacing $\zeta$ with $-\zeta$, and both sides of the equation are multiplied by $\sigma_{3}$,
$$
\sigma_{3}J_{x}(-\zeta)\sigma_{3}+\mathrm{i} \zeta^{2}\left[\sigma_{3},\sigma_{3}J(-\zeta)\sigma_{3}\right]=-\zeta \sigma_{3}Q J(-\zeta)\sigma_{3},
$$
due to $\sigma_{3} Q \sigma_{3}=-Q$, the above equation can be reduced to
$$
\sigma_{3}J_{x}(-\zeta)\sigma_{3}+\mathrm{i} \zeta^{2}\left[\sigma_{3},\sigma_{3}J(-\zeta)\sigma_{3}\right]=\zeta Q\sigma_{3} J(-\zeta)\sigma_{3}.
$$
It is easy to find that $J(\zeta)$ and $J(-\zeta)$ satisfy the same equation, so there is
\begin{equation}J(\zeta)=\sigma_{3} J(-\zeta) \sigma_{3},\end{equation}
it follows that
\begin{equation}
\mathbf{P}(-\zeta)=\sigma_{3} \mathbf{P}(\zeta) \sigma_{3},
\end{equation}
and
\begin{equation}
S(-\zeta)=\sigma_{3} S(\zeta) \sigma_{3}.
\end{equation}
\end{proof}

From the (\ref{2.30})and (\ref{2.33}), we obtain the relations
\begin{equation}
s_{11}^{*}\left(\zeta^{*}\right)=\hat{s}_{11}(\zeta), \quad s_{21}^{*}\left(\zeta^{*}\right)=\hat{s}_{12}(\zeta), \quad s_{12}^{*}\left(\zeta^{*}\right)=\hat{s}_{21}(\zeta),  \quad \zeta \in \mathbb{R} \cup i \mathbb{R},\label{2.34}
\end{equation}
and
\begin{equation}
s_{1 1}(\zeta)=s_{1 1}(-\zeta), ~~s_{2 2}(\zeta)=s_{2 2}(-\zeta),~~
s_{1 2}(-\zeta)=-s_{1 2}(\zeta), ~~~s_{2 1}(-\zeta)=-s_{2 1}(\zeta),\quad \zeta \in \mathbb{R} \cup i \mathbb{R}.\label{2.37}
\end{equation}
Thus $s_{11}(\lambda)$ is an even function, and each zero $\zeta_{k}$ of $s_{11}$ is accompanied with zero $-\zeta_{k}$. Similarly, $\hat{s}_{11}$ has two zeros $\pm\zeta_{k}^{*}$.

\subsection{Solvability of RHP problem}\

In general, if the $\operatorname{det} \mathbf{P}(\zeta)\neq 0$ of the RHP, the RHP is considered to be regular, its solution is unique, and can be given by using Plemelj formula\cite{YJK-2010-SIAM}. But more often than not they are non-regular, where $\operatorname{det} \mathbf{P}(\zeta)= 0$, i.e,$s_{11}(\pm\zeta_k)=0$ and $\hat{s}_{11}(\pm\bar{\zeta}_k)=0$ at certain discrete locations, $\pm\zeta_k$ and $\pm\bar{\zeta_k}$ are called zeros. Here we first consider the case of simple zeros $\{\pm\zeta_{k} \in \mathbb{C}_{13},1 \leq k \leq N\}$ and $\{\pm\bar{\zeta}_{k} \in \mathbb{C}_{24}, 1 \leq k \leq N\},$ where $N$ is the number of these zeros. These zeros are known from the relation (\ref{2.34}) above
$$
s_{11}(\zeta_k)=\hat{s}_{11}^{*}(\zeta^{*}_k)=0,~~\hat{s}_{11}(\bar{\zeta}_k)=0,
$$
so there are
\begin{equation}\bar{\zeta}_{k}=\zeta_{k}^{*}. \label{2.38}\end{equation}

In this case, both $\operatorname{ker}\left(\mathbf{P}\left(\pm \zeta_{k}\right)\right)$ are one-dimensional and spanned by single column vector $\left|v_{k}\right\rangle$ and single row vector $\left\langle v_{k}\right|,$ respectively, thus
\begin{equation}\mathbf{P}\left(\zeta_{k}\right) \left|v_{k}\right\rangle=0, \quad \left\langle v_{k}\right| \mathbf{P}\left(\zeta_{k}^{*}\right)=0, \quad\zeta_{k}\in \mathbb{C}_{13}, \quad 1 \leq k \leq N.\label{2.40}\end{equation}
By the symmetry relation (\ref{2.29}), it is easy to get
\begin{equation}\left|v_{k}\right\rangle=\left\langle v_{k}\right|^{\dagger}.\label{2.39}\end{equation}

Regarding this non-regular RHP (\ref{RHP}) under the canonical normalization condition, its solution is also unique. Next we construct a matrix function $\Gamma(x,t,\zeta)$ which could cancel all the zeros of $\mathbf{P}$.  From the relations (\ref{2.34}) and (\ref{2.37}), we should construct a matrix $\Gamma_{k}$ whose determinant is
\begin{equation}\operatorname{det} \Gamma_{k}(\zeta)=\frac{\zeta^{2}-\zeta_{k}^{2}}{\zeta^{2}-\zeta_k^{* 2}}.\label{L6}
\end{equation}
From the above properties (\ref{2.29}), (\ref{2.32}) and (\ref{L6}), we could readily construct the explicit form
for the matrix
\begin{equation}
\Gamma_{k}(\zeta)=I+\frac{A_{k}}{\zeta-\zeta_{k}^{*}}-\frac{\sigma_{3} A_{k} \sigma_{3}}{\zeta+\zeta_{k}^{*}}, \quad \Gamma_{k}^{-1}(\zeta)=I+\frac{A_{k}^{\dagger}}{\zeta-\zeta_{k}}-\frac{\sigma_{3} A_{k}^{\dagger} \sigma_{3}}{\zeta+\zeta_{k}}, \quad \zeta_{k}\in \mathbb{C}_{13}, \quad k=1,2, \ldots, N\label{L7}
\end{equation}
where
\begin{equation}
A_{k}=\frac{\zeta_{k}^{* 2}-\zeta_{k}^{2}}{2}\left(\begin{array}{cc}
\alpha_{k}^{*} & 0 \\
0 & \alpha_{k}
\end{array}\right)\left|w_{k}\right\rangle\left\langle w_{k}\right|,~~~
\alpha_{k}^{-1}=\langle w_{k}|\left(\begin{array}{cc}
\zeta_{k} & 0 \\
0 & \zeta_{k}^{*}
\end{array}\right)| w_{k}\rangle,
\end{equation}
\begin{equation}
\left|w_{k}\right\rangle=\Gamma_{k-1}\left(\zeta_{k}\right) \cdots \Gamma_{1}\left(\zeta_{k}\right)\left|v_{k}\right\rangle, \quad\left\langle w_{k}|=| w_{k}\right\rangle^{\dagger},\label{L10}
\end{equation}
then det $\mathbf{P}\Gamma^{-1}_{k}\neq 0$ at points $\pm \zeta_{k}$ and det $\Gamma^{-1}_{k}\mathbf{P}\neq 0$ at points $\pm \zeta_{k}^{*}.$ Introducing
\begin{equation}
\Gamma(\zeta)=\Gamma_{N}(\zeta) \Gamma_{N-1}(\zeta) \cdots \Gamma_{1}(\zeta),\label{L8}
\end{equation}
\begin{equation}
\Gamma^{-1}(\zeta)=\Gamma_{1}^{-1}(\zeta)\Gamma_{2}^{-1}(\zeta) \cdots \Gamma_{N}^{-1}(\zeta),\label{L9}
\end{equation}
and the analytic solutions may be represented as
\begin{equation}
\mathbf{P}=\tilde{\mathbf{P}}\Gamma .\label{2.45}
\end{equation}
Therefore, $\Gamma(x,t,\zeta)$ accumulates all zero of the RHP, and then  the RHP of the higher-order KN equation without singularity is
\begin{rhp}
The matrix function $\mathbf{\tilde{P}}(\zeta; x)$ has the following properties:

$\bullet$ $\mathbf{Analyticity}: \mathbf{\tilde{P}}(\zeta; x, t)$ is analytic function in $\zeta\in\mathbb{C}_{13}\cup\mathbb{C}_{24}$;

$\bullet$$\mathbf{Jump\quad Condition}$: \begin{equation} \mathbf{\tilde{P}}_{+}(\zeta; x)=\mathbf{\tilde{P}}_{-}(\zeta;x)\Gamma G \Gamma^{-1}(\zeta),~~\quad \zeta \in \mathbb{R} \cup \mathrm{i} \mathbb{R}.\label{RHP1}
    \end{equation}

$\bullet$ $\mathbf{Asymptotic behaviors}: ~~~~~\mathbf{\tilde{P}}(\zeta; x)=\mathbf{\tilde{P}}_{0}+O(\zeta),~~~~~\text{as } \zeta\rightarrow 0.$
\end{rhp}
The form of $G$ has been given by equation (\ref{G1}). From Eq.(\ref{2.45}) we have
\begin{equation}
\mathbf{\tilde{P}}_{0}=(\left.\Gamma\right|_{\zeta=0})^{-1}.\label{p0}
\end{equation}

\subsection{ Scattering data evolution}\

From the solutions of the RHP, we see that the scattering data needed
to solve this RHP and reconstruct the potential are
$$
\{s_{21}, \hat{s}_{12}, \zeta \in \mathbb{R} \cup \mathrm{i} \mathbb{R}; \pm \zeta_{k}, \pm \zeta_{k}^{*}, |v_{k}\rangle, \langle v_{k}|,1\leq k\leq N\}.
$$
Since $J$ satisfies the temporal equation (\ref{Jt}) of the Lax pair and the relation (\ref{2.21}), then according to the evolution property (\ref{2.21}) and $Q \rightarrow 0, V\rightarrow 0$ as $|x| \rightarrow \infty,$ we have
$$
S_{t}+2 \mathrm{i} \zeta^{6}\left[\sigma_{3}, S\right]=0,\label{2.41}
$$
the evolutions of the entries of the scattering matrix $S$ satisfy
\begin{equation}
s_{11, t}=s_{22, t}=0,~~
s_{12}(t ; \zeta)=s_{12}(0 ; \zeta) \exp \left(-4 \mathrm{i} \zeta^{6} t\right), \quad s_{21}(t ; \zeta)=s_{21}(0 ; \zeta) \exp \left(4 \mathrm{i} \zeta^{6} t\right).\label{L16}
\end{equation}

Differentiating both sides of the first equation of (\ref{2.40}) with respect to $x$ and $t,$ and recalling the Lax (\ref{Jx})-(\ref{Jt}) we have
$$
\mathbf{P}(\zeta_{k} ; x)\left(\frac{d|v_{k}\rangle}{d x}+\mathrm{i} \zeta^{2} \sigma_{3}|v_{k}\rangle\right)=0, \quad \mathbf{P}(\zeta_{k} ; x)\left(\frac{d|v_{k}\rangle}{d t}+2 \mathrm{i} \zeta^{6} \sigma_{3}|v_{k}\rangle\right)=0,~~~\zeta_{k}\in \mathbb{C}_{13}.
$$
It concludes that
$$
\left|v_{k}\right\rangle=e^{-\mathrm{i} \zeta_{k}^{2} \sigma_{3} x-2 \mathrm{i} \zeta_{k}^{6} \sigma_{3} t}\left|v_{k0}\right\rangle,
$$
where $v_{k0}=\left.v_{k}\right|_{x=0}$.

\section{N Soliton Solutions}
In this part, we mainly obtain the potential $q$. The expansion of $\mathbf{P}(\zeta)$  with $\zeta \rightarrow 0$,
\begin{equation}\mathbf{P}(\zeta)=I+\mathbf{P}^{(1)} \zeta+\mathbf{P}^{(2)} \zeta^{2}+O\left(\zeta^{2}\right).\label{3.1}
\end{equation}
Substituting the expansion into Eq. (\ref{Jx}), the potential matrix function can be obtained by comparing the coefficients
\begin{equation}
Q=\mathbf{P}_{x}^{(1)},\label{2.43}
\end{equation}
from this formula, we can get the potential $q(x,t)$. It is well known that the soliton solutions correspond to the vanishing of scattering coefficients, $G=I, \hat{G}=0$. Thus, we intend to solve the corresponding RHP(\ref{RHP1}).

According to equations (\ref{2.45}) and (\ref{p0}), we can consider the following expansion form
\begin{equation}\mathbf{P}(x, t ; \zeta)=(\Gamma|_{\zeta=0})^{-1}(\Gamma|_{\zeta=0}+\Gamma^{(1)}(x, t)\zeta+O(\zeta)),\end{equation}
which gives $ \mathbf{P}^{(1)}=(\Gamma|_{\zeta=0})^{-1}\Gamma^{(1)}(x, t)$. Below, the main effort is to find an explicit expression for
$(\Gamma|_{\zeta=0})^{-1}\Gamma^{1}(x, t).$ In fact, the form of $\Gamma$ from (\ref{L8})and (\ref{L9}) have more compact form
\begin{equation}
\Gamma(\zeta)=I+\sum_{k=1}^{N}\left[\frac{B_{k}}{\zeta-\zeta_{k}^{*}}-\frac{\sigma_{3} B_{k} \sigma_{3}}{\zeta+\zeta_{k}^{*}}\right],\label{gamma0}
\end{equation}
and
\begin{equation}
\Gamma^{-1}(\zeta)=I+\sum_{k=1}^{N}\left[\frac {B_{k}^{\dagger} }{\zeta-\zeta_{k}}-\frac{\sigma_{3}  B_{k}^{\dagger} \sigma_{3}}{\zeta+\zeta_{k}}\right],\label{gamma2}
\end{equation}
with $B_{k}=\left|z_{k}\right\rangle\langle v_{k}|$. To determine the form of matrix $B_{k}$, we consider $\Gamma(\zeta)\Gamma^{-1}(\zeta)=I$, we have
$$
\operatorname{Res}_{\zeta=\zeta_{j}}\Gamma(\zeta)\Gamma^{-1}(\zeta)=\Gamma(\zeta_j)B_{j}^{\dagger} =0,
$$
and it yields
\begin{equation}
\left[I+\sum_{k=1}^{N}\left(\frac{\left|z_{k}\right\rangle\left\langle v_{k}\right| }{\zeta_{j}-\zeta_{k}^{*}}-\frac{\sigma_{3}\left|z_{k}\right\rangle\left\langle v_{k}\right|  \sigma_{3}}{\zeta_{j}+\zeta_{k}^{*}}\right)\right]\left|v_{j}\right\rangle=0, \quad j=1,2, \ldots N
\end{equation}
it's easy to figure out
\begin{equation}
\left|z_{k}\right\rangle_{1} =\sum_{j=1}^{N} (M^{-1})_{jk}\left|v_{j}\right\rangle_{1},
\end{equation}
where $\left|z_{k}\right\rangle_{l}$ denotes the $l-$th element of $\left|z_{k}\right\rangle$, matrix $M$ is defined as
\begin{equation}
M_{jk}=\frac{\left\langle v_{k}\left|\sigma_{3}\right| v_{j}\right\rangle}{\zeta_{j}+\zeta_{k}^{*}}-\frac{\langle v_{k}\mid v_{j}\rangle}{\zeta_{j}-\zeta_{k}^{*}}.\label{m}
\end{equation}

Then we have
$$
(\Gamma|_{\zeta=0})=I-\sum_{j=1}^{N}\left[\frac{B_{j}+\sigma_{3} B_{j} \sigma_{3}}{\zeta_{j}^{*}}\right],
$$

$$
\Gamma^{(1)}(x,t)=-\sum_{j=1}^{N} \frac{B_{j}-\sigma_{3} B_{j} \sigma_{3}}{\zeta_{j}^{*2}}.
$$
From these equations enable us to have
$$
\mathbf{P}^{(1)}=(\Gamma|_{\zeta=0})^{-1}\Gamma^{(1)}(x,t)=\left(I-\sum_{j=1}^{N}\left[\frac{B_{j}+\sigma_{3} B_{j} \sigma_{3}}{\zeta_{j}^{*}}\right]\right)^{-1}\sum_{j=1}^{N} \frac{\sigma_{3} B_{j} \sigma_{3}-B_{j}}{\zeta_{j}^{*2}},
$$
by Eq. (\ref{2.43}), we can obtain that the potential function $q(x,t)$ is
\begin{equation}
q(x,t)=\left((1-\sum_{j,k=1}^{N}\frac{2(M^{-1})_{jk}|v_{k}\rangle_{1}\langle v_{j}|_{1}}{\zeta_{j}^{*}})^{-1}(\sum_{j,k=1}^{N}\frac{-2(M^{-1})_{jk}|v_{k}\rangle_{1}\langle v_{j}|_{2}}{\zeta_{j}^{*2}})\right)_{x},\label{soliton}
\end{equation}
where $M$ has been given by Eq. (\ref{m}). Notice that $M^{-1}$ can be expressed as the transpose of $M's$ cofactor matrix divided by det$M$.
Hence the solution (\ref{soliton}) can be rewritten as
\begin{equation}
q(x,t)=\left(\frac{2\frac{detF}{detM}}{1+2\frac{detG}{detM}}\right)_{x}=\left(\frac{2detF}{detM+2detG}\right)_x,\label{qsoliton}
\end{equation}
where
$$
F=\left(\begin{array}{ccccc}
M_{11} & M_{12} & \cdots & M_{1 n} & |v_{1}\rangle_1 \\
M_{21} & M_{22} & \cdots & M_{2 n} & |v_{2}\rangle_1 \\
\vdots & \vdots & \ddots & \vdots & \vdots \\
M_{{n}1} & M_{{n} 2} & \cdots & M_{nn} & |v_{n}\rangle_1  \\
\frac{\langle v_{1}|_2}{{\zeta_1}^{*2}} & \frac{\langle v_{2}|_2}{{\zeta_2}^{*2}} & \cdots & \frac{\langle v_{n}|_2}{{\zeta_n}^{*2}} & 0
\end{array}\right),~~~~~
G=\left(\begin{array}{ccccc}
M_{11} & M_{12} & \cdots & M_{1 n} & |v_{1}\rangle_1 \\
M_{21} & M_{22} & \cdots & M_{2 n} & |v_{2}\rangle_1 \\
\vdots & \vdots & \ddots & \vdots & \vdots \\
M_{{n}1} & M_{{n} 2} & \cdots & M_{nn} & |v_{n}\rangle_1  \\
\frac{\langle v_{1}|_1}{{\zeta_1}^{*}} & \frac{\langle v_{2}|_1}{{\zeta_2}^{*}} & \cdots & \frac{\langle v_{n}|_1}{{\zeta_n}^{*}} & 0
\end{array}\right).
$$
To get the explicit $N$-soliton solutions, we may take $v_{k0}=(a_k,b_k)^{T}$, then
$$
\left|v_{k}\right\rangle=\left(\begin{array}{c}
a_{k}e^{\theta_{k}} \\
b_{k}e ^{-\theta_{k}}\\
\end{array}\right),~~~~~\left\langle v_{k}\right|=(\begin{array}{cc}
a_{k}^{*}e^{\theta^{*}_{k}} &b_{k}^{*}e ^{-\theta^{*}_{k}}
\end{array}),
$$
where $\theta_{k}=-i \zeta_{k}^2 x-2 i \zeta_{k}^{6}t$.

In what follows, we will take single soliton and two-soliton solutions as examples to study the properties of solitons in more detail. For convenience, let $\zeta_{j}=\zeta_{jR}+i \zeta_{jI},$
$$
\begin{array}{l}
\theta_{jR}=2 m_{j}(x-(8 m_j^{2}-6 \beta_{j}^2) t), ~~~~~\theta_{jI}=-\beta_{j} x-2(\beta_{j}^{3}-12 m_{j}^{2}\beta_j) t,\\
m_{j}=\zeta_{jR} \zeta_{jI}, \quad \beta_{j}=\zeta_{jR}^2-\zeta_{jI}^{2},
\end{array}
$$
where $\zeta_{jR} ,\zeta_{jI}$ are the real and imaginary parts of $\zeta_{j}$.

\subsection{ Single-soliton solution}\

For $N=1$, taking the discrete spectrum point $\pm\zeta_1$ and $\pm\zeta_1^{*}$, then using the formula (\ref{soliton}) to directly calculate, it can be seen
\begin{equation}
q(x, t)=\frac{(\zeta_{1}^{2}-\zeta_{1}^{*2})}{|\zeta_{1}|^{2} } \left( \frac{a_{1} b_{1}^{*} e^{\theta_{1}-\theta_{1}^{*}}}{\zeta_{1}^{*}|b_{1}|^{2} e^{-\left(\theta_{1}+\theta_{1}^{*}\right)}+\zeta_{1}|a_{1}|^{2} e^{\theta_{1}+\theta_{1}^{*}}}\right)_{x},\label{L21}
\end{equation}
or equal to
\begin{equation}
q(x,t)=8a_1b_1^{*}\zeta_{1R}\zeta_{1I}\frac{\zeta_{1}|b_1|^2e^{- 2\theta_{1R}}+\zeta_{1}^{*}|a_1|^2e^{ 2\theta_{1R}}}{(\zeta_{1}^{*}|b_1|^2e^{- 2\theta_{1R}}+\zeta_{1}|a_1|^2e^{ 2\theta_{1R}})^2}e^{2i\theta_{1I}},\label{L32}
\end{equation}


The velocity for the single soliton is $v =8 \zeta_{1R}^2\zeta_{1I}^2-6 (\zeta_{1R}^2-\zeta_{1I}^2)^2$. The center position
for $|q|$ locates on the line
 $$x-vt-\frac{1}{4 \zeta_{1R}^2\zeta_{1I}^2}ln\frac{|b_1|}{|a_1|}=0.$$
 The amplitudes associated with $|q|^2$ are given by
 $$
 A(q)=\frac{64\zeta_{1R}^2\zeta_{1I}^2}{2|\zeta_{1}^2|+\zeta_{1}^2+\zeta_{1}^{*2}}.
 $$
In Fig. (\ref{f2})(a), we give the 3-D graph of the single-soliton solution.

\subsection{ Two-soliton solutions}\

When $N = 2$, the solution (\ref{soliton}) can also be written out. The two-soliton solutions of higher-order KN equation has the
form of $q(x, t) = \Delta_1/\Delta_0$ with
$$
\begin{array}{l}
\Delta_1=\delta_1e^{-2\theta_{1R}+2i\theta_{1I}-4\theta_{2R}}+\delta_2e^{2\theta_{1R}+2i\theta_{1I}}+\delta_3e^{2\theta_{1R}+2i\theta_{1I}
-4\theta_{2R}}+\delta_4e^{-2\theta_{1R}+2i\theta_{1I}}
+\delta_5e^{-2\theta_{2R}+2i\theta_{2I}}+\delta_6e^{4i\theta_{1I}-2\theta_{2R}-2i\theta_{2I}}\\
~~~+\delta_7e^{4\theta_{1R}+2\theta_{2R}+2i\theta_{2I}}+\delta_8e^{4\theta_{1R}-2\theta_{2R}+2i\theta_{2I}}+\delta_9e^{2\theta_{2R}+2i\theta_{2I}}
+\delta_{10}e^{2\theta_{1R}-2i\theta_{1I}+4i\theta_{2I}}
+\delta_{11}e^{2\theta_{1R}+2i\theta_{1I}+4\theta_{2R}}+\delta_{12}e^{-2\theta_{1R}+2i\theta_{1I}+4\theta_{2R}}\\
~~~+\delta_{13}e^{4i\theta_{1I}+2\theta_{2R}-2i\theta_{2I}}+\delta_{14}e^{-4\theta_{1R}-2\theta_{2R}+2i\theta_{2I}}+\delta_{15}e^{-4\theta_{1R}
+2\theta_{2R}+2i\theta_{2I}}+\delta_{16}e^{-2\theta_{1R}-2i\theta_{1I}+4i\theta_{2I}},
\end{array}
$$
$$
\begin{array}{l}
\Delta_0=\rho_0+\rho_1e^{-4\theta_{1R}-4\theta_{2R}}+\rho_2e^{-4\theta_{1R}}+\rho_3e^{-4\theta_{2R}}+\rho_4e^{-2\theta_{1R}-2i\theta_{1I}-2\theta_{2R}+2i\theta_{2I}}
+\rho_5e^{-2\theta_{1R}+2i\theta_{1I}-2\theta_{2R}-2i\theta_{2I}}+\rho_6e^{4\theta_{1R}+4\theta_{2R}}\\
~~~+\rho_7e^{4\theta_{1R}}+\rho_8e^{4\theta_{2R}}+\rho_9e^{2\theta_{1R}-2i\theta_{1I}-2\theta_{2R}+2i\theta_{2I}}+\rho_{10}e^{2\theta_{1R}+2i\theta_{1I}+2\theta_{2R}-2i\theta_{2I}}
+\rho_{11}e^{4\theta_{1R}-4\theta_{2R}}+\rho_{12}e^{2\theta_{1R}-2i\theta_{1I}+2\theta_{2R}+2i\theta_{2I}}\\
~~~+\rho_{13}e^{2\theta_{1R}+2i\theta_{1I}-2\theta_{2R}-2i\theta_{2I}}+\rho_{14}e^{-4\theta_{1R}+4\theta_{2R}}+\rho_{15}e^{-2\theta_{1R}-2i\theta_{1I}+2\theta_{2R}+2i\theta_{2I}}+\rho_{16}e^{-2\theta_{1R}+2i\theta_{1I}+2\theta_{2R}-2i\theta_{2I}}
+\rho_{17}e^{-4i\theta_{1I}+4i\theta_{2I}}\\
~~~+\rho_{18}e^{4i\theta_{1I}-4i\theta_{2I}}.
\end{array}
$$

The coefficients of these exponential terms are constituted of $a_1, a_1^{*},a_2, a_2^{*},b_1, b_1^{*},b_2, b_2^{*}$ and $\zeta_{1},\zeta_{1}^{*},\zeta_{2},\zeta_{2}^{*}$.
However, it is tedious to write them all out, and they can be calculated directly via the computer. Instead of presenting the complex expression, we show the typical solution behavior in Fig. (\ref{f2})(b). It can be seen from Fig. (\ref{f2})(b) that when $t\rightarrow-\infty$, the solution consists of two single solitons that are far apart and travel opposite each other. When they collide together, the interaction weakens. When $t\rightarrow+\infty$, these are separated into two single solitons, and there is no change in shape and speed, and no energy radiation is emitted to the far field. Therefore, the interaction of these solitons is elastic. But it can be observed from the graph that after the interaction, each soliton has a phase shift and a position shift.
\begin{figure}[htpb]
\centering
\includegraphics[width=4.5cm,height=3.5cm]{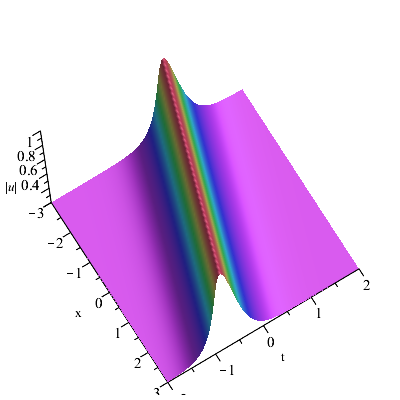}\qquad \qquad
\includegraphics[width=3.5cm,height=3.3cm]{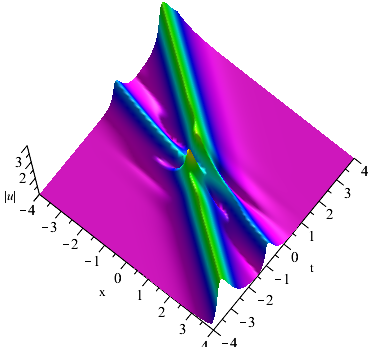}\\
$(a)$ \qquad\qquad \qquad \qquad \qquad$(b)$
\caption{(Color online) Soliton solutions for $|q|$.(a) one-soliton solution in three-dimensional plot.where $\zeta_1=1+0.25i,a_1=1,b_1=0.1+0.7i$,(b) two-soliton solutions in three-dimensional plot. where $\zeta_1=1+0.25i,a_1=1,b_1=0.1+0.7i,\zeta_2=1+0.5i,a_2=1,b_2=-0.1+0.7i.$}\label{f2}
\end{figure}

Next, we verify the rationality of the above analysis through the expression of the soliton solution. In general, making the assumption $\xi_i\eta_i>0$ and $v_{1}<v_{2}$. This means that at $t\rightarrow-\infty$, soliton-2 is on the left side of soliton-1 and moves faster, and the two solitons are in the moving frame with velocity $v_{i}=8 \zeta_{iR}^2\zeta_{iI}^2-6 (\zeta_{iR}^2-\zeta_{iI}^2)^2 $, note that $\theta_{1R}=2 m_{1}(x-v_{1} t),\theta_{2R}=2 m_{2}(x-v_{2}t)$, it yields
$$
m_2\theta_{1R}-m_1\theta_{2R}=2m_1m_2(v_{2}-v_{1})t,
$$
we used the asymptotic analysis technique \cite{XT-2010-JPA}, we intend to investigate the collision dynamics of these two-soliton solutions. Then we have asymptotic expressions of $q(x,t)$ under different asymptotic states of $\theta_{1R}$ and $\theta_{2R}$.

(i) Before collision (as $t\rightarrow-\infty$),

(a) If $|\theta_{1R}|<\infty$, then $ \theta_{2R}\rightarrow\infty$:
\begin{equation}
q(x, t) \sim  8\tilde{a}_1^{M}\tilde{b}_1^{*M}\zeta_{1R}\zeta_{1I} \frac{\zeta_{1}|\tilde{b}_1^{M}|^2e^{- 2\theta_{1R}}+\zeta_{1}^{*}|\tilde{a}_1^{M}|^2e^{ 2\theta_{1R}}}{(\zeta_{1}^{*}|\tilde{b}_1^{M}|^2e^{- 2\theta_{1R}}+\zeta_{1}|\tilde{a}_1^{M}|^2e^{ 2\theta_{1R}})^2}e^{2i\theta_{1I}},
\end{equation}
where $\tilde{a}_1^{M}=a_1(\zeta_{2}^2-\zeta_{1}^2), \tilde{b}_1^{M}=b_1(\zeta_{2}^{*2}-\zeta_{1}^2)$.

(b) If $|\theta_{2R}|<\infty$, then $\theta_{1R} \rightarrow-\infty$:
\begin{equation}
q(x, t) \sim  8\tilde{a}_2^{M}\tilde{b}_2^{*M}\zeta_{2R}\zeta_{2I} \frac{\zeta_{2}|\tilde{b}_2^{M}|^2e^{- 2\theta_{2R}}+\zeta_{2}^{*}|\tilde{a}_2^{M}|^2e^{ 2\theta_{2R}}}{(\zeta_{2}^{*}|\tilde{b}_2^{M}|^2e^{- 2\theta_{2R}}+\zeta_{2}|\tilde{a}_2^{M}|^2e^{ 2\theta_{2R}})^2}e^{2i\theta_{2I}},
\end{equation}
where $\tilde{a}_2^{M}=a_2(\zeta_{2}^{2}-\zeta_{1}^{*2}), \tilde{b}_2^{M}=b_2(\zeta_{2}^2-\zeta_{1}^2)$.

(ii) After collision (as $t\rightarrow\infty$),

(a) If $|\theta_{1R}|<\infty$, then $\theta_{2R} \rightarrow-\infty$:
\begin{equation}
q(x, t) \sim  8\tilde{a}_1^{P}\tilde{b}_1^{*P}\zeta_{1R}\zeta_{1I} \frac{\zeta_{1}|\tilde{b}_1^{P}|^2e^{- 2\theta_{1R}}+\zeta_{1}^{*}|\tilde{a}_1^{P}|^2e^{ 2\theta_{1R}}}{(\zeta_{1}^{*}|\tilde{b}_1^{P}|^2e^{- 2\theta_{1R}}+\zeta_{1}|\tilde{a}_1^{P}|^2e^{ 2\theta_{1R}})^2}e^{2i\theta_{1I}},
\end{equation}
where $\tilde{a}_1^{P}=a_1(\zeta_{2}^{*2}-\zeta_{1}^2), \tilde{b}_1^{P}=b_1(\zeta_{2}^{2}-\zeta_{1}^2)$.

(b) If $|\theta_{2R}|<\infty$, then $\theta_{1R} \rightarrow\infty$:
\begin{equation}
q(x, t) \sim  8\tilde{a}_2^{P}\tilde{b}_2^{*P}\zeta_{2R}\zeta_{2I} \frac{\zeta_{2}|\tilde{b}_2^{P}|^2e^{- 2\theta_{2R}}+\zeta_{2}^{*}|\tilde{a}_2^{P}|^2e^{ 2\theta_{2R}}}{(\zeta_{2}^{*}|\tilde{b}_2^{P}|^2e^{- 2\theta_{2R}}+\zeta_{1}|\tilde{a}_2^{P}|^2e^{ 2\theta_{2R}})^2}e^{2i\theta_{2I}},
\end{equation}
where $\tilde{a}_2^{P}=a_2(\zeta_{2}^{2}-\zeta_{1}^2), \tilde{b}_2^{P}=b_2(\zeta_{2}^{2}-\zeta_{1}^{*2})$.

It is pointed out that the asymptotic solutions can also be written as the function of solitary waves, and the respective velocities are $v_1$ and $v_2$, which remain unchanged before and after the collision. This elastic interaction is a remarkable property, which indicates that DNLS Eq.(\ref{hDNLS1}) is integrable. From the above asymptotic solution, we can get the phase difference of soliton-1 solution,

$$
\Delta\theta_{01}=\frac{1}{2}\left(\ln\frac{ |\tilde{b}_{1}^{P}|}{|\tilde{a}_{1}^{P}|}-\ln \frac{|\tilde{b}_{1}^{M}|}{|\tilde{a}_{1}^{M}|}\right)= \ln \left|\frac{\zeta_{2}^2-\zeta_{1}^{2}}{\zeta_{2}^{*2}-\zeta_{1}^{2}}\right|.
$$
Following similar calculations, we can get the phase difference of soliton-2 solution,
$$
\Delta \theta_{02}=\frac{1}{2}\left(\ln\frac{|\tilde{b}_{2}^{P}|}{|\tilde{a}_{2}^{P}|}-\ln \frac{|\tilde{b}_{2}^{M}|}{|\tilde{a}_{2}^{M}|}\right)= \ln \left|\frac{\zeta_{2}^2-\zeta_{1}^{*2}}{\zeta_{2}^{2}-\zeta_{1}^{2}}\right|=-\Delta\theta_{01}.
$$

 \section{Soliton matrix for multiple zeros}
In this section, we will further consider the case of multiple zeros, where the multiplicity of $\{\pm\zeta_{i},\pm\zeta^{*}_{i}\}$ is greater than 1, then the determinant of $\mathbf{P}$ can be written in the following form:
$$
\operatorname{det}\mathbf{P}_{+}(\zeta)=\left(\zeta^{2}-\zeta_{1}^{2}\right)^{n_{1}}\left(\zeta^{2}-\zeta_{2}^{2}\right)^{n_{2}} \cdots\left(\zeta^{2}-\zeta_{r}^{2}\right)^{n_{r}} \rho(\zeta),
$$
$$
\operatorname{det}\mathbf{P}_{-}^{-1}(\zeta)=\left(\zeta^{2}-\zeta_{1}^{*2}\right)^{n_{1}}\left(\zeta^{2}-\zeta_{2}^{*2}\right)^{n_{2}} \cdots\left(\zeta^{2}-\zeta_{r}^{*2}\right)^{n_{r}} \hat{\rho}(\zeta),
$$
 where $\rho(\zeta_i)\neq 0$ $(i=1..r)$ for all $\zeta\in \mathbb{C}_{13}$, and $\hat{\rho}(\zeta_i)\neq 0$ $(i=1..r)$ for all $\zeta\in \mathbb{C}_{24}$.

Compared with the case of simple zeros, the number of kernel functions with multiple zeros is related to the multiplicity of zeros. For example, for discrete spectral point $\{\zeta_1,\zeta_1^{*}\}$, its kernel function is
\begin{equation}\mathbf{P}\left(\zeta_{1}\right) \left|v_{j}\right\rangle=0, \quad \left\langle v_{j}\right| \mathbf{P}\left(\zeta_{1}^{*}\right)=0, \quad\zeta_{1}\in \mathbb{C}_{13}, \quad 1 \leq j \leq n_1,\end{equation}
$\left|v_{j}\right\rangle$ is linearly independent. For the case of multiple zeros, the corresponding $\Gamma$ and $\Gamma^{-1}$
can be given by using the following theorem,

\begin{lemma}(\cite{VSS-2003-SAM},Lemma 3)
 Consider a pair of higher order zeros of order $n_j$ $(j=1,..,r)$: $\{\zeta_j,-\zeta_j\}$ in $\mathbb{C}_{13}$ and $\{\zeta_j^{*},-\zeta_j^{*}\}$ in $\mathbb{C}_{24}$. Then the corresponding soliton matrix $\Gamma_j(\zeta)$  and its inverse can be cast in the following form
 \begin{equation}
 \begin{array}{l}
\Gamma^{-1}_j(\zeta)=I+\left(\left|\phi_{j,1}\right\rangle, \cdots,\left|\tilde{{\phi}}_{j,n_j}\right\rangle\right) \Xi_j(\zeta)\left(\begin{array}{c}
\left\langle \tilde{\varphi}_{j,n_j}\right| \\
\vdots \\
\left\langle{\varphi}_{j,1}\right|
\end{array}\right), \\
\Gamma_j(\zeta)=I+\left(\left|\bar{\varphi}_{j,n_j}\right\rangle, \cdots,\left|\bar{\tilde{\varphi}}_{j,1}\right\rangle\right) \overline{\Xi_j}(\zeta)\left(\begin{array}{c}
\left\langle\bar{\phi}_{j,1}\right| \\
\vdots \\
\left\langle{\bar{{\tilde{\phi}}}_{j,n_j}}\right|
\end{array}\right),
\end{array}\label{5.16}
 \end{equation}
 where the matrices $\Xi_j(\zeta)$ and $\overline{\Xi}_j(\zeta)$ are defined as
 $$
\Xi_j(\zeta)=\left(\begin{array}{cc}
\mathcal{D}^{+}\left(\zeta-\zeta_{j}\right) & \mathbf{0}_{n \times n} \\
\mathbf{0}_{n \times n} & \mathcal{D}^{+}\left(\zeta+\zeta_{j}\right)
\end{array}\right), \quad \overline{\Xi_j}(\zeta)=\left(\begin{array}{cc}
\mathcal{D}^{-}\left(\zeta-\zeta_{j}^{*}\right) & \mathbf{0}_{n \times n} \\
\mathbf{0}_{n \times n} & \mathcal{D}^{-}\left(\zeta+\zeta_{j}^{*}\right)
\end{array}\right),
$$
$\mathcal{D}^{+}(\gamma),\mathcal{D}^{-}(\gamma)$ are upper-triangular and lower-triangular Toeplitz matrices defined as:
$$
\mathcal{D}^{+}(\gamma)=\left(\begin{array}{cccc}
\gamma^{-1} & \gamma^{-2} & \cdots & \gamma^{-n} \\
0 & \ddots & \ddots & \vdots \\
\vdots & \ddots & \gamma^{-1} & \gamma^{-2} \\
0 & \cdots & 0 & \gamma^{-1}
\end{array}\right),~~~~~~\mathcal{D}^{-}(\gamma)=\left(\begin{array}{cccc}
\gamma^{-1} & 0 & \cdots & 0 \\
\gamma^{-2} & \gamma^{-1} & \ddots & \vdots \\
\vdots & \ddots & \ddots & 0 \\
\gamma^{-n} & \cdots & \gamma^{-2} & \gamma^{-1}
\end{array}\right),
$$
and vectors $|\phi_{j,i}\rangle,|\tilde{\phi}_{j,i}\rangle,\langle\bar{\tilde{\phi}}_{j,i}|,\langle\varphi_{j,i}|,
\langle\tilde{\varphi}_{j,i}|,|\bar{\tilde{\varphi}}_{j,i}\rangle(i=1,..,n_N)$ are independent of $\zeta$.
\end{lemma}

So there are
\begin{equation}
\Gamma(\zeta)=\Gamma_{r}(\zeta) \Gamma_{r-1}(\zeta) \cdots \Gamma_{1}(\zeta),
\end{equation}
\begin{equation}
\Gamma^{-1}(\zeta)=\Gamma_{1}^{-1}(\zeta)\Gamma_{2}^{-1}(\zeta) \cdots \Gamma_{r}^{-1}(\zeta).
\end{equation}
In fact, the rest of the vector parameters in (\ref{5.16}) can be derived by calculating the  residue of each order in the identity $\Gamma(\zeta) \Gamma^{-1}(\zeta)=I$ at $\zeta=\zeta_{j}$ and $\zeta=-\zeta_{j}$,
$$
\Gamma\left(\zeta_{j}\right)\left(\begin{array}{c}
\left|\phi_{j,1}\right\rangle \\
\vdots \\
\left|\phi_{j,n_r}\right\rangle
\end{array}\right)=0, \quad \Gamma\left(-\zeta_{j}\right)\left(\begin{array}{c}
\left|\tilde{\phi}_{j,1}\right\rangle \\
\vdots \\
\left|\tilde{\phi}_{j,n_r}\right\rangle
\end{array}\right)=0,
$$
where
$$
\Gamma(\zeta)=\left(\begin{array}{cccc}
\Gamma(\zeta) & 0 & \cdots & 0 \\
\frac{d}{d \zeta} \Gamma(\zeta) & \Gamma(\zeta) & \ddots & \vdots \\
\vdots & \ddots & \ddots & 0 \\
\frac{1}{(n_r-1) !} \frac{d^{n_r-1}}{d \zeta^{n_r-1}} \Gamma(\zeta) & \cdots & \frac{d}{d\zeta} \Gamma(\zeta) & \Gamma(\zeta)
\end{array}\right).
$$
Using this method, the process of solving soliton solution is very complex. Next, the corresponding $\Gamma$ can be constructed by using the method of \cite{GBL-2012-SAM}, the dressing matrix of multiple zeros is derived by unipolar limit method. The results are given by the following theorem:
\begin{theorem}
Suppose $\zeta=\zeta_i$ is the zero of geometric multiplicity $n_j(j=1,..,r)$, and $\sum_{j=1}^{r}n_j=N$, then the modified matrix can be expressed as
$$
\Gamma(\zeta)=\Gamma_{r}^{[n_r-1]} \cdots \Gamma_{r}^{[0]}\cdots\Gamma_{1}^{[n_1-1]} \cdots \Gamma_{1}^{[0]}, \quad \Gamma^{-1}(\zeta)=(\Gamma_{1}^{[0]})^{-1} \cdots (\Gamma_{1}^{[n_1-1]})^{-1}\cdots(\Gamma_{r}^{[0]})^{-1} \cdots (\Gamma_{r}^{[n_r-1]})^{-1},
$$
where
$$
\Gamma_{i}^{[j]}(\zeta)=I+\frac{A_{i}^{[j]}}{\zeta-\zeta_i^{*}}-\frac{\sigma_{3} A_{i}^{[j]} \sigma_{3}}{\zeta+\zeta_i^{*}}, \quad (\Gamma_{i}^{[j]})^{-1}(\zeta)=I+\frac{A_{i}^{\dagger[j]}}{\zeta-\zeta_{i}}-\frac{\sigma_{3} A_{i}^{\dagger[j]} \sigma_{3}}{\zeta+\zeta_{i}},
$$
$$
A_{i}^{[j]}=\frac{\zeta_{i}^{*2}-\zeta_{i}^{2}}{2}\left(\begin{array}{cc}
\alpha_{i}^{*[j]} & 0 \\
0 & \alpha_{i}^{[j]}
\end{array}\right)|v_{i}^{[j]}\rangle\langle v_{i}^{[j]}|,~~~(\alpha_{i}^{[j]})^{-1}=\langle v_{i}^{[j]}|\left(\begin{array}{cc}
\zeta_i & 0 \\
0 & \zeta_i^{*}
\end{array}\right)|v_{i}^{[j]}\rangle,
$$
$$
|v_{i}^{[j]}\rangle=\lim_{\delta \rightarrow 0} \frac{(\Gamma_{i}^{[n_j-1]} \cdots \Gamma_{i}^{[0]}\cdots\Gamma_{1}^{[n_1-1]} \cdots \Gamma_{1}^{[0]})|_{\zeta=\zeta_{i}+\delta}}{\delta^{j}}|v_{i}\rangle(\zeta_{i}+\delta),
$$
$$
\langle v_{i}^{[j]}|=\lim_{\delta \rightarrow 0}\langle v_{1}|(\zeta_{i}^{*}+\delta) \frac{(\Gamma_{1}^{[0]-1} \cdots \Gamma_{1}^{[n_1-1]-1} \cdots\Gamma_{i}^{[0]-1} \cdots \Gamma_{i}^{[n_j-1]-1})|_{\zeta=\zeta_{i}^{*}+\delta}}{\delta^{j}}.
$$
\end{theorem}

Then we can get
$$
P^{(1)}=\left(I-\sum_{i=1}^{r}\sum_{j=0}^{n_i-1}\left[\frac{B_{i}^{[j]}-\sigma_{3} B_{i}^{[j]} \sigma_{3}}{\zeta_{i}^{*}}\right]\right)^{-1}\sum_{i=1}^{r}\sum_{j=0}^{n_i-1} \frac{\sigma_{3} B_{i}^{[j]} \sigma_{3}-B_{i}^{[j]}}{\zeta_{i}^{*2}},
$$
which lead to
\begin{equation}
q(x,t)=\left(\frac{2\frac{det\tilde{F}}{det\tilde{M}}}{1+2\frac{det\tilde{G}}{det\tilde{M}}}\right)_{x}
=\left(\frac{2det\tilde{F}}{det\tilde{M}+2det\tilde{G}}\right)_x,\label{qsoliton2}
\end{equation}
where
$$
\tilde{F}=\left(\begin{array}{ccccc}
\tilde{M}^{[11]} & \tilde{M}^{[12]} & \cdots & \tilde{M}^{[1 r]} & \tilde{\chi}_1 \\
\tilde{M}^{[21]} & \tilde{M}^{[22]} & \cdots & \tilde{M}^{[2 r]} & \tilde{\chi}_2 \\
\vdots & \vdots & \ddots & \vdots & \vdots \\
\tilde{M}^{[r 1]} & \tilde{M}^{[r 2]} & \cdots & \tilde{M}^{[r r]} & \tilde{\chi}_r \\
\tilde{\psi}_1 & \tilde{\psi}_2 & \cdots & \tilde{\psi}_r & 0
\end{array}\right),~~~~
\tilde{M}=\left(\begin{array}{ccccc}
\tilde{M}^{[11]} & \tilde{M}^{[12]} & \cdots & \tilde{M}^{[1 r]}  \\
\tilde{M}^{[21]} & \tilde{M}^{[22]} & \cdots & \tilde{M}^{[2 r]}  \\
\vdots & \vdots & \ddots & \vdots  \\
\tilde{M}^{[r 1]} & \tilde{M}^{[r 2]} & \cdots & \tilde{M}^{[r r]}
\end{array}\right),
$$
$$
\tilde{G}=\left(\begin{array}{ccccc}
\tilde{M}^{[11]} & \tilde{M}^{[12]} & \cdots & \tilde{M}^{[1 r]} & \tilde{\chi_1} \\
\tilde{M}^{[21]} & \tilde{M}^{[22]} & \cdots & \tilde{M}^{[2 r]} & \tilde{\chi_2} \\
\vdots & \vdots & \ddots & \vdots & \vdots \\
\tilde{M}^{[r 1]} & \tilde{M}^{[r 2]} & \cdots & \tilde{M}^{[r r]} & \tilde{\chi_r} \\
\tilde{\tau}_1 & \tilde{\tau}_2 & \cdots & \tilde{\tau}_r & 0
\end{array}\right),
$$
with
$$
\tilde{M}_{kl}^{[ij]}=\frac{1}{(k-1) !(l-1) !} \frac{\partial^{k+l-2}}{\partial \zeta^{*k-1} \partial \zeta^{l-1}}\frac{\left\langle v_{j}\left|\sigma_{3}\right| v_{i}\right\rangle}{\zeta+\zeta^{*}}-\frac{\langle v_{j}\mid v_{i}\rangle}{\zeta-\zeta^{*}}|_{\zeta=\zeta_i,\zeta^{*}=\zeta_j^{*}}.
$$

$$
\tilde{\chi}_i=\left(|v_{i}\rangle_1^{[0]}, |v_{i}\rangle_1^{[1]}, \quad \cdots, \quad |v_{i}\rangle_1^{[n_i-1]} \right)^{T},~~~~~~~~~~~~|v_{i}\rangle^{[j]}=\frac{1}{(j)!}\frac{\partial^{j}}{\partial(\zeta)^{j}}|v_{i}\rangle|_{\zeta=\zeta_{i}},
$$
$$
\tilde{\psi}_i=\left((\frac{\langle v_{i}|_2}{{\zeta_i}^{*2}})^{[0]}, (\frac{\langle v_{i}|_2}{{\zeta_i}^{*2}})^{[1]}, \quad \cdots, \quad (\frac{\langle v_{i}|_2}{{\zeta_1}^{*2}})^{[n_i-1]} \right),~~~~~~(\frac{\langle v_{i}|}{{\zeta_i}^{*2}})^{[j]}=(\frac{1}{(j)!}\frac{\partial^{j}}{\partial(\zeta)^{j}}(\frac{\langle v_{i}|}{{\zeta}^{*2}})|_{\zeta=\zeta_{i}^{*}}),
$$
$$
\tilde{\tau}_i=\left((\frac{\langle v_{i}|_1}{{\zeta_i}^{*}})^{[0]}, (\frac{\langle v_{i}|_1}{{\zeta_i}^{*}})^{[1]}, \quad \cdots, \quad (\frac{\langle v_{i}|_1}{{\zeta_i}^{*}})^{[n_i-1]} \right),~~~~~~(\frac{\langle v_{i}|}{{\zeta_i}^{*}})^{[j]}=(\frac{1}{(j)!}\frac{\partial^{j}}{\partial(\zeta)^{j}}(\frac{\langle v_{i}|}{{\zeta}^{*}})|_{\zeta=\zeta_{i}^{*}}).
$$

Hence, formula (\ref{qsoliton2}) gives the general expression of high-order solitons with multiple zeros. Because the spectral parameters here cannot be pure real or pure virtual, the expression of high-order soliton is relatively complex, but different $n_j$ and appropriate parameters can be selected, and the graphics of mixed high-order solitons solution can be given by using mathematical software such as Maple and Mathematica.
Here, we give several representative mixed solutions. In Fig. (\ref{f3}), let $n_1 = 2,n_j=0 (j=2..r)$, in (\ref{qsoliton2}) which represent the simple double-zero case, and $n_1 = 3,n_j=0(j=2..r)$, in (\ref{qsoliton2}) is the simple triple-zero case. In Fig. (\ref{f4}), take $n_1 = 2,n_2=1,n_j=0 (j=3..r)$, that is, a mixed solution of a double zero and a single zero, and take $n_1 = 2,n_2=2,n_j=0(j=3..r)$, which means a mixed solution of two double zeros.

\begin{figure}[htpb]
\centering
\includegraphics[width=4.5cm,height=3.5cm]{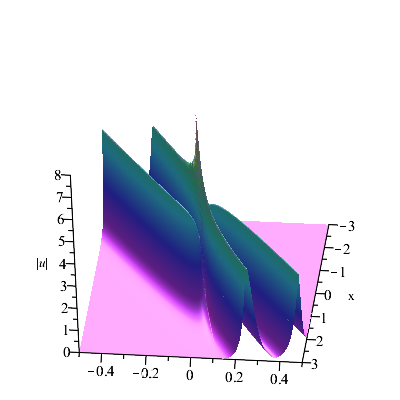}
\includegraphics[width=3.5cm,height=3.3cm]{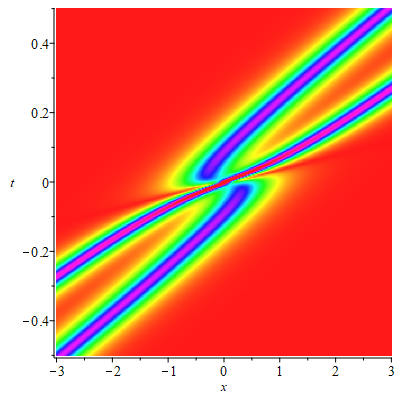}
\includegraphics[width=4.5cm,height=3.5cm]{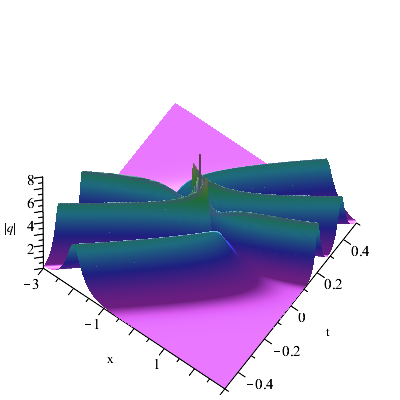}
\includegraphics[width=3.5cm,height=3.0cm]{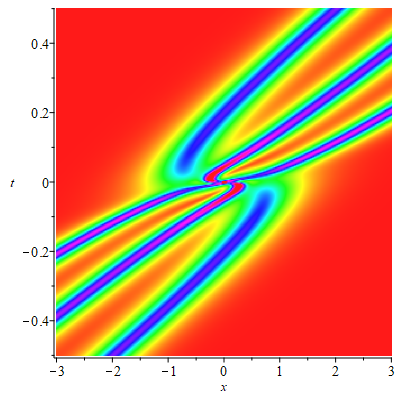}\\
~~~~$(a)$\qquad \qquad \qquad \qquad\qquad \qquad\qquad$(b)$\qquad \qquad\qquad \qquad \qquad\qquad $(c)$\qquad \qquad\qquad \qquad \qquad \qquad$(d)$
\caption{(Color online)(a) The double zeros soliton solutions for $|q|$ and $n_1=2,\zeta=1+i,a_1=b_1=1$.(b) Density plot of double zeros.(c)The triple zeros soliton solutions for $|q|$ and $n_1=2,\zeta=1+i,a_1=b_1=1$. (d) Density plot of triple zeros.}\label{f3}
\end{figure}

\begin{figure}[htpb]
\centering
\includegraphics[width=4.5cm,height=3.5cm]{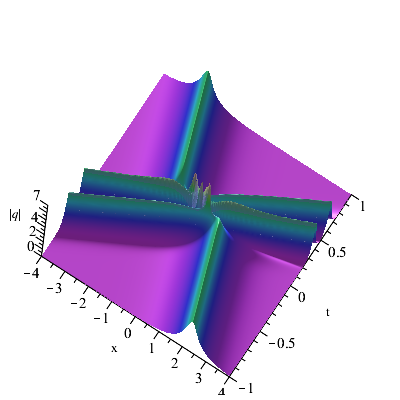}
\includegraphics[width=3.5cm,height=3.3cm]{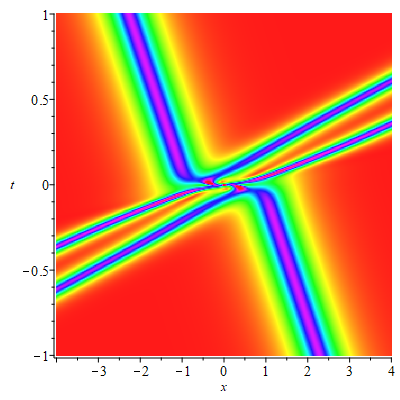}
\includegraphics[width=4.5cm,height=3.5cm]{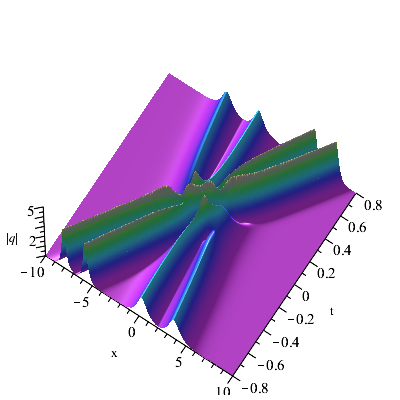}
\includegraphics[width=3.5cm,height=3.0cm]{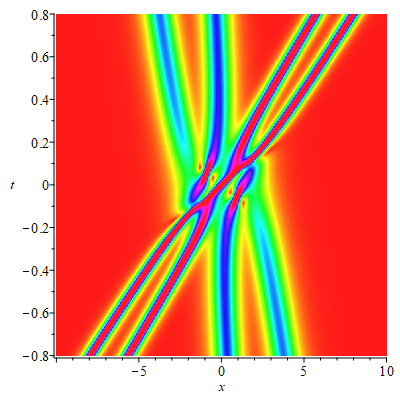}\\
$(a)$\qquad \qquad \qquad \qquad\qquad \qquad \qquad \qquad$(b)$\qquad \qquad\qquad \qquad \qquad $(c)$\qquad  \qquad \qquad \qquad \qquad$(d)$
\caption{(Color online)(a) Mixed solutions of double zeros and single zero for $|q|$ and $n_1=2,n_2=1,\zeta_1=1+i,a_1=b_1=1,\zeta_2=\frac{1}{2}+i,a_2=b_2=1$. (b)Density plot of single-double-zeros. (c) Mixed solutions of double zeros and double zeros for $|q|$ and $n_1=2,n_2=2,\zeta_1=1+i,a_1=b_1=1,\zeta_2=1+\frac{1}{2}i,a_2=b_2=1$. (d) Density plot of double-double  zeros.}\label{f4}
\end{figure}

\section{ Conclusions  and discussions}
In a word, the inverse scattering method is applied to the higher-order KN equation with vanishing boundary, and the soliton matrix is constructed by studying the corresponding RHP. Using RHP regularization of finite simple zeros, the determinant form of general N-solitons of higher-order KN equation without reflection is obtained, which is different from the soliton solution form of previous KN system. In the process of inverse scattering, the potential function is recovered when the spectral parameter tends to zero, which effectively avoids the appearance of implicit function\cite{JP-2021-ar}. At the same time, the properties of the single soliton solution and the collision dynamics and asymptotic behavior of the two soliton solution are investigated.

In addition, the multiple zeros of RHP are considered, and the higher-order soliton matrix of higher-order KN equation is obtained by limit technique. Several typical graphs are given, including the graphs of double zeros solutions, trip zeros solutions,  single-double zeros solutions, and double-double zeros solutions. It provides a good basis for future experimental observation.

In this context, we only consider the solution with vanishing boundary conditions. How to adjust the analysis to find the Jost solution of the spectral problem, so as to obtain the solution with non vanishing boundary conditions. The global well posedness, long-time behavior and asymptotic stability of solitons need to be further studied.

\end{document}